\documentclass[conference]{IEEEtran}
\IEEEoverridecommandlockouts
\usepackage{cite}
\usepackage{amsmath,amssymb,amsfonts}
\usepackage{amsthm}
\usepackage{algorithmic}
\usepackage{graphicx}
\usepackage{textcomp}
\usepackage{xcolor}
\def\BibTeX{{\rm B\kern-.05em{\sc i\kern-.025em b}\kern-.08em
    T\kern-.1667em\lower.7ex\hbox{E}\kern-.125emX}}

\DeclareMathOperator\supp{Supp}
    
\begin{document}
\title{Achievable Rates and Algorithms for Group Testing with Runlength Constraints
}

\author{\IEEEauthorblockN{Stefano Della Fiore}
\IEEEauthorblockA{\textit{Deparment of Information Engineering} \\
\textit{University {of Brescia}}\\
Brescia (BS), Italy \\
s.dellafiore001@unibs.it}
\and
\IEEEauthorblockN{Marco Dalai}
\IEEEauthorblockA{\textit{Deparment of Information Engineering} \\
\textit{University {of Brescia}}\\
Brescia (BS), Italy \\
marco.dalai@unibs.it}
\and
\IEEEauthorblockN{Ugo Vaccaro}
\IEEEauthorblockA{\textit{Department of Informatics} \\
\textit{University {of Salerno}}\\
Fisciano (SA), Italy \\
uvaccaro@unisa.it}
}

\newtheorem{thm}{Theorem}[section]
\newtheorem{lem}[thm]{Lemma}
\newtheorem{cor}[thm]{Corollary}
\newtheorem{prop}[thm]{Proposition}
\newtheorem{conj}[thm]{Conjecture}
\newtheorem{prob}[thm]{Problem}
\newtheorem{ques}[thm]{Question}
\newtheorem{ex}[thm]{Exercise}
\newtheorem{rem}[thm]{Remark}

\theoremstyle{definition} 
\newtheorem{defn}[thm]{Definition} 
\newtheorem{exa}[thm]{Example}

\maketitle

\begin{abstract}
In this paper, we study bounds on the minimum length of $(k,n,d)$-superimposed codes introduced by Agarwal {\em et al.} \cite{Olgica}, in the context of Non-Adaptive Group Testing algorithms with runlength constraints. A $(k,n,d)$-superimposed code of length $t$ is a $t \times n$ binary matrix such that any two 1's  in each column are separated by a run of at least $d$ 0's, and such that for any column $\mathbf{c}$ and any other $k-1$ columns, there exists a row where $\mathbf{c}$ has $1$ and all the remaining $k-1$ columns have $0$. Agarwal {\em et al.} proved the existence of such codes with $t=\Theta(dk\log(n/k)+k^2\log(n/k))$. Here we investigate more in detail the coefficients in front of these two main terms as well as the role of lower order terms. We show that improvements can be obtained over the construction  in \cite{Olgica} by using different constructions and by an appropriate exploitation of the Lov\'asz Local Lemma in this context. Our findings also suggest $O(n^k)$ randomized Las Vegas algorithms for the construction of such codes. We also extend our results to Two-Stage Group Testing algorithms with runlength constraints.
\end{abstract}


\section{Introduction}
Group Testing refers to the scenario in which one has a population $I$ of individuals, and an unknown subset $P$ of $I$, commonly referred to as ``positives''. The goal is to determine the unknown elements of $P$ by performing tests on arbitrary subsets $A$ of $I$ (called {\em pools}), and the outcome of the test is assumed to return the value 1 (positive) if $A$ contains at least one element of the unknown set $P$, the value 0 (negative), otherwise. The problem was first introduced by Dorfman \cite{Dorf} during WWII, in the context of mass blood testing. Since then, Group Testing techniques have found applications in a large variety of areas, ranging from DNA sequencing to quality control, data security to network analysis, and much more. We refer the reader to the excellent monographs \cite{DuHwang,JSA} for an account of the vast literature on the subject.

 Group Testing procedures can be adaptive or non-adaptive. In adaptive Group Testing, the tests are performed sequentially, and the content of the pool tested  at the generic step $i$ might depend on the previous $i-1$ test outcomes.
 Conversely, in non-adaptive Group
Testing all pools are a-priori set, and tests are carried out in parallel. Non-adaptive Group Testing (NAGT) schemes typically require  more tests to discover
the positives, but they are faster since tests can be performed in parallel. To combine the advantages of both techniques, while
mitigating their limitations, it is sometimes preferable to implement a hybrid approach, where a first  screening is performed via a NAGT 
 algorithm, followed by a simple one-by-one testing of  the  members that are identified in the first stage as
potentially positives. This latter approach is usually called Two-Stage Group Testing \cite{Vaccaro2}. 

In NAGT, the algorithm to determine the positives is usually represented by means of a $t\times n$ binary matrix $M$, where each row of $M$ represents a test while each column is associated to a distinct member of the population $I=\{1, 2, \ldots , n\}$. More precisely, we have 
$M_{ij}=1$ if and only if the member $j\in I$ belongs to the $i$-th test.  In general, one assumes a known upper bound $k$ on the cardinality of the unknown set of positives $P$. Having said that, the property one usually requires for $M$ to represent a correct (and efficiently decodable) NAGT is the following \cite{DuHwang}:
for any $k$-tuple of the $n$ columns of $M$ we demand  that for any column $\mathbf{c}$ of the given $k$-tuple, there exists a row ${i \in \{1,\ldots , t\}}$ such that $\mathbf{c}$ has symbol $1$ in row $i$ and all the remaining $k-1$ columns of the $k$-tuple have a $0$ in the same row $i$. This condition renders matrices $M$ with such a property equivalent to the well known superimposed codes introduced in the seminal paper by Kautz and Singleton \cite{KS} and, independently, by Erd\"os {\em et al.} in \cite{Erdos}.

Motivated by applications in topological DNA-based data storage, the authors of \cite{Olgica} introduced an interesting new variant of NAGT, in which the associated test matrix $M$ has to satisfy additional constraints, in order to comply with the biological constraints of the problem they want to solve. Informally, one of the main problems studied in \cite{Olgica} is to  show the existence of a superimposed code $M$ with a "small" number $t$ of rows and satisfying the following additional property: any two 1's  in each column are  separated by a run of  at least $d$ 0's. We refer the reader to \cite{Olgica} for the rationale behind this run-length constraint. The main achievability result obtained in \cite{Olgica} says that codes with these properties exist for $t=\Theta(dk\log(n/k)+k^2\log(n/k))$.

{\bf Our results.} We study the achievable coefficients in front of the $dk\log(n/k)$ and $k^2\log(n/k)$ terms and on lower order terms in $k$, whose values are of importance as they determine the achievable rates of such codes for fixed values of $d$ and $k$. We show that better results than those derived in \cite{Olgica} can be obtained using different random coding constructions which also admit simpler analyses. Also, we show that improved results can be obtained by an appropriate use of the  Lov\'asz Local Lemma (see, e.g., \cite{AlonSpencer}). By exploiting the celebrated result by Moser and Tardos \cite{MT}, this directly implies a $O(n^k)$ randomized Las Vegas  algorithm to construct such codes. 
In the final part of this paper, we also extend our results to Two Stages Group Testing algorithms.

\section{New Upper Bounds}\label{sec:superimposed}

Throughout the paper, the logarithms without subscripts are in base two, and we denote with $\ln (\cdot)$ the natural logarithm.  For notation convenience we denote with $[a,  b]$ the set ${\{a, a+1, \ldots, b\}}$.

\begin{defn}\cite{Olgica}
Let $k$, $n$, $d$ be positive integers,  $k \leq n$.
A $(k, n,d)$-\emph{superimposed} code is a $t \times n$ binary matrix $M$ such that any two 1's  in each column of $M$  are  separated by a run of  at least $d$ 0's,  and for any $k$-tuple of the columns of $M$ we have that for any column $\mathbf{c}$ of the given $k$-tuple, there exists a row ${i \in [1,t]}$ such that $\mathbf{c}$ has symbol $1$ in row $i$ and all the remaining $k-1$ columns of the $k$-tuple are equal to $0$.  The number of rows $t$ of $M$ is called the length of the $(k,n,d)$-superimposed code.
\end{defn}

\begin{defn}
A $(k,n,d,w)$-\emph{superimposed} code is a $(k,n,d)$-superimposed with the additional constraint that each column has weight $w$ (number of ones).
\end{defn}

First,  we need the following enumerative lemma.

\begin{lem}\label{lem:enumerative}
Let $S \subseteq \{0, 1\}^t$ be the set of all distinct binary vectors of length $t$ such that each vector  has Hamming weight $w \geq 1$ and  any two 1's  in each vector  are  separated by a run of  at least $d$ 0's.  If $t \geq (w-1)d + w$,  then
$$
	|S| = \binom{t - (w-1) d}{w}\,.
$$
\end{lem}
\begin{proof}
Let $A$ be the set of all distinct binary vectors of length $t-(w-1)d$ and weight $w$.
One can see that $|S| = |A|$ since each element of $S$ can be obtained from an element $a \in A$ by adding between each pair of  consecutive ones in $a$ exactly $d$ 0's.  Conversely,  each element of $A$ can be obtained from an element $s \in S$ by removing between each pair of consecutive ones in $s$  exactly $d$ 0's. 
\end{proof}

We also need the following technical lemma and an easy corollary, which have been proved in \cite{Vaccaro1}.  We report here the proofs for the reader's convenience.

\begin{lem}\label{lem:firstBoundBinomial}
Let $a, b, c$ be positive integers such that $c \leq a \leq b$. We have that
$$
	\frac{a}{b} \cdot \frac{a-c}{b-c} \leq \left(\frac{a - \frac{c}{2}}{b-\frac{c}{2}}\right)^2\,.
$$
\end{lem}
\begin{proof}
Clearly $a (a-c) c^2 \leq b (b-c) c^2$. Then adding the quantity $4ab (a-c)(b-c)$ to both members implies that $a(a-c) (2b-c)^2 \leq b(b-c) (2a-c)^2$. Therefore Lemma~\ref{lem:firstBoundBinomial} follows.
\end{proof}

\begin{cor}\label{cor:boundBinomial}
Let $a, b, c$ be positive integers such that $ c \leq a \leq b$.  We have that
\begin{equation}\label{eq:boundB}
	\frac{\binom{a}{c}}{\binom{b}{c}} \leq \left(\frac{a - \frac{c-1}{2}}{b - \frac{c-1}{2}}\right)^{c}\,.
\end{equation}
\end{cor}
\begin{proof}
Expanding the LHS of \eqref{eq:boundB} we get
\begin{equation}\label{eq:expBin}
	\frac{\binom{a}{c}}{\binom{b}{c}} = \frac{a}{b} \cdot \frac{a-1}{b-1} \cdots \frac{a-c+1}{b-c+1}\,.
\end{equation}
Let us group the terms in \eqref{eq:expBin} into pairs as follows
\begin{equation}\label{eq:termsBin}
	\frac{a-i}{b-i} \cdot \frac{a- (c-i-1)}{b-(c-i-1)} \text{ for } i=0,\ldots, \left\lceil \frac{c-1}{2} \right\rceil -1 \,.
\end{equation}
If $c$ is odd then we leave alone the term $(a-\frac{c-1}{2})/(b-\frac{c-1}{2})$.  By Lemma \ref{lem:firstBoundBinomial},  each term in \eqref{eq:termsBin} can be upper bounded by
$$
	\frac{a-i}{b-i} \cdot \frac{a- (c-i-1)}{b-(c-i-1)} \leq \left(\frac{a - \frac{c-1}{2}}{b-\frac{c-1}{2}}\right)^2\,.
$$
Hence Corollary \ref{cor:boundBinomial} follows.
\end{proof}

The main tool to prove Theorem \ref{th:LLExistence} is the Lov\'asz Local Lemma for the symmetric case.  We state here the lemma.

\begin{lem} \cite{AlonSpencer}\label{lem:LLL}
Let $E_1,  E_2,  \ldots,  E_m$ be events in an arbitrary probability space.  Suppose that each event $E_i$ is mutually independent of a set of all other events $E_j$ but at most $d$, and that $\Pr(E_i) \leq p$ for all $1 \leq i \leq m$. If
$$
	e d p \leq 1
$$
then $\Pr(\cap_{i=1}^m \overline{E}_i) > 0$.
\end{lem}

Now, we are ready to state our main result.

\begin{thm}\label{th:LLExistence}
There exists a $(k,n,d,w)$-superimposed code of length $t$,  where $t$ is the minimum integer such that the following inequality holds
\begin{equation}\label{eq:standBound}
e  k \left[ \binom{n}{k} - \binom{n-k+1}{k} \right] \left( \frac{w (k-1) - \frac{w-1}{2}}{t - (w-1) d - \frac{w-1}{2}} \right)^w \leq 1.
\end{equation}
\end{thm}

\begin{proof}
Let $M$ be a $t \times n$ binary matrix, where each column $\mathbf{c}$ is picked uniformly at random between the set of all distinct binary vectors of length $t$ such that each column has weight $w$ and  any two 1's  in each column of $M$  are  separated by a run of  at least $d$ 0's.  Therefore by Lemma \ref{lem:enumerative} we have that
$$\Pr(\mathbf{c}) = \binom{t - (w-1) d}{w}^{-1}\,.$$
For a given index $i \in [1, n]$ and a set of column-indices $B$, $|B| = k-1$, $i \not \in B$,  let $E_{i, B}$ be the event such that for every row in which $\mathbf{c}_i$ (the $i$-th column) has $1$,  there exists an index $j \in B$ such that $\mathbf{c}_j$ has $1$ in that same row. We can write this event in terms of supports as $\supp(\mathbf{c}_i)\subseteq \supp(\mathbf{c}_B)$.  There are $n \binom{n-1}{k-1}$ such events.  We can express the probability of such an event as follows
\begin{align}\label{eq:prob}
	\Pr(E_{i, B}) = 
	&\sum_{c' = (c'_1, \ldots, c'_{k-1})}  \Pr(\mathbf{c}_B = c') \cdot \nonumber \\ & \Pr(\supp(\mathbf{c}_i)\subseteq \supp(\mathbf{c}_B) |  \mathbf{c}_B = c'),
\end{align}
where we have denoted with $\mathbf{c}_B$ the vector $(\mathbf{c}_{j_1}, \ldots, \mathbf{c}_{j_{k-1}})$ in which  $j_1, \ldots,  j_{k-1}$ are the elements of $B$.  The sum in \eqref{eq:prob} is over all the possible configurations of $k-1$ vectors of length $t$, weight $w$ and the distance between ones in each column is at least $d$.  Then, we can upper bound \eqref{eq:prob} by the maximum of $\Pr(\supp(\mathbf{c}_i)\subseteq \supp(\mathbf{c}_B) |  \mathbf{c}_B = c')   $ over all $k-1$ vectors $c'= (c'_1, \ldots, c'_{k-1})$. 
Therefore, we can consider the worst-case scenario where the $k-1$ columns of $M$ with indices in $B$ maximize this probability.  It can be seen that the maximum is achieved when the $w (k-1)$ ones of the $k-1$ columns indexed by $B$ are placed in $w (k-1)$ different rows.
Hence,
\begin{equation}\label{eq:upperbound}
\Pr(E_{i, B}) \leq \frac{\binom{w(k-1)}{w}}{\binom{t-(w-1)d}{w}} \, .
\end{equation}
Using Corollary \ref{cor:boundBinomial} we upper bound \eqref{eq:upperbound} as follows
\begin{equation}\label{eq:upperboundSimp}
\Pr(E_{i, B}) \leq  \left( \frac{w (k-1) - \frac{w-1}{2}}{t - (w-1) d - \frac{w-1}{2}} \right)^w \, .
\end{equation}
Proceeding as in \cite{Vaccaro1},  it can be proved that an arbitrary event $E_{i, A}$ is mutually independent from all the events $E_{j, C}$, where $C \subseteq [1, n] \setminus (A \cup \{i\})$ and $j \not \in C$.  Since the number of events $E_{j,C}$ is equal to
$$
\binom{n-k}{k-1} (n-k+1) = k \binom{n-k+1}{k}\,,
$$
each event $E_{i,A}$ is dependent of at most 
\begin{equation}\label{eq:D}
f = k \left[ \binom{n}{k} - \binom{n-k+1}{k} \right]
\end{equation}
other events. 
If the probability that none of the events $E_{i,A}$ occurs is strictly positive then there exists a matrix $M$ that is a $(k,n,d,w)$-superimposed code of length $t$.
Therefore, using Lemma \ref{lem:LLL} and taking $p$ equal to the RHS of \eqref{eq:upperboundSimp} and $f$ as defined in equation \eqref{eq:D},  Theorem \ref{th:LLExistence} follows.
\end{proof}

\begin{rem}\label{rem:UnionBound}
We note that in Theorem \ref{th:LLExistence} we could use the union bound instead of the Local Lemma.  Since the total number of events is $n \binom{n-1}{k-1}$, we have that there exists a $(k,n,d,w)$-superimposed code of length $t$, provided that
\begin{equation}\label{eq:UnionBound}
	n \binom{n-1}{k-1} \left( \frac{w (k-1) - \frac{w-1}{2}}{t - (w-1) d - \frac{w-1}{2}} \right)^w < 1\,.
\end{equation}
\end{rem}

In \cite{Olgica}  the authors  proved that a $(k,n,d,w)$-superimposed code of length $t$ exists, provided that 
\begin{equation}\label{eq:OlgicaBound}
n \binom{n-1}{k-1} \left( \frac{w(k-1)}{t - (2d+1)(w-1)} \right)^w < 1\,.
\end{equation}

It is clear that our  bound given in Remark \ref{rem:UnionBound} is better than the bound given in  \eqref{eq:OlgicaBound} since
$$
	 \frac{w (k-1) - \frac{w-1}{2}}{t - (w-1) d - \frac{w-1}{2}} \leq \frac{w(k-1)}{t - (2d+1)(w-1)}
$$
for all positive integers $w,  k, d$.  

If we compare the bounds of Theorem \ref{th:LLExistence} and Remark \ref{rem:UnionBound} then it has been proved in \cite{Vaccaro1} that
\begin{equation}\label{eq:BoundUgoK}
	e k \left[ \binom{n}{k} - \binom{n-k+1}{k} \right] \leq n \binom{n-1}{k-1}
\end{equation}
for all $k \leq 0.667 \sqrt{n}$. Therefore when $k$ is much smaller than $n$ (which is indeed the case in circumstances of interest), the Local Lemma performs better than the union bound.  It is important to note that a conjecture of Erd\H{o}s, Frankl and F\"uredi \cite{Erdos} says that for $k \geq \sqrt{n}$ optimal superimposed codes have length equal to $n$.  The current best known result has been proved in \cite{Shann} which states that if $k \geq 1.157 \sqrt{n}$ then the minimum length of superimposed codes is equal to $n$. 

\begin{cor}\label{cor:impBound}
There exists a $(k,n,d,w)$-superimposed code of length $t$, where
\begin{multline}\label{eq:boundwithW}
	t \leq \Bigg\lceil (w-1) d +  \frac{w-1}{2} + \left(w(k-1) - \frac{w-1}{2}\right) \cdot  \\ \left(  \min\left\{ n \binom{n-1}{k-1},  e k \left[ \binom{n}{k} - \binom{n-k+1}{k} \right] \right\} \right)^{\frac{1}{w}}  \Bigg\rceil \,.
\end{multline}
\begin{proof}
It easily follows rearranging the terms in equation \eqref{eq:standBound} and in equation \eqref{eq:UnionBound}.
\end{proof}
\end{cor}

\begin{cor}\label{cor:impBound2}
There exists a $(k,n,d)$-superimposed code of length $t$ with $k \leq n / e$, where
\begin{multline*}
	t \leq \ln 2 \cdot  dk  \log (n/k)  + e^2 \cdot k^2 \log (n/k) \\ - \frac{\left( 3e^2 - \ln 2\right)}{2} k \log (n/k) - d + O(1) \,.
\end{multline*}
\begin{proof}
Substitute $w = k \ln (n / k)$ in \eqref{eq:boundwithW} and upper bound
$$
	\min\left\{ n \binom{n-1}{k-1},  e k \left[\binom{n}{k} - \binom{n-k+1}{k} \right] \right\} < k \left( \frac{e n}{k} \right)^k\,.
$$
Therefore we obtain
\begin{multline}\label{eq:explicitBound}
	t \leq d \left(k \ln (n/k) - 1\right) + \frac{k}{2} \ln(n/k) \\ + e \cdot (k e^k)^{\frac{1}{k\ln(n/k)}}  k \left(k -\frac{3}{2} \right) \ln (n/k)  + O(1).
\end{multline}
Hence Corollary \ref{cor:impBound2} follows since $n \geq e k$ and $k^{1/k} \leq \frac{1}{\ln 2}$ for every integer $k \geq 1$.
\end{proof}
\end{cor}
We note that in the explicit bound given in Corollary \ref{cor:impBound2} the leading coefficient of the term $k \log (n/k)$ can be improved, for $k \leq 0.667 \sqrt{n}$, by using a better estimation of the minimum in equation \eqref{eq:boundwithW} that comes from the use of the Local Lemma.

By exploiting the celebrated result by Moser and Tardos \cite{MT}, this directly implies a 
$O(n^k)$ randomized Las Vegas algorithm to construct the codes of Corollary \ref{cor:impBound2}

\smallskip

From the inequality \eqref{eq:OlgicaBound} we can derive an explicit upper bound on the length of the codes 
whose existence was showed in  \cite{Olgica}  when $w = k \ln (n/k)$ by upper bounding  $n \binom{n-1}{k-1}$ with  $k\left(\frac{en}{k}\right)^k$.
We report here the obtained result.

\begin{thm}\cite{Olgica}\label{th:OlgicaThm}
There exists a $(k,n,d)$-superimposed code of length $t$, where
\begin{multline*}
	t \leq 2 d\left(k \ln (n/k) - 1\right) + k \ln(n/k) \\ +  e \cdot (k e^k)^{\frac{1}{k\ln(n/k)}}  k (k-1) \ln (n/k) + O(1).
\end{multline*}
\end{thm}

It is clear that our result given in equation \eqref{eq:explicitBound} improves the one of Theorem \ref{th:OlgicaThm}.

\begin{rem} \label{rem:NonExistence}
We note that it was proved in \cite{Olgica} that every $(k, n,d)$-superimposed code of length $t$ must satisfy
$$
   t \geq \min \left\{ n,  1 + (k-1)(d+1) \right\}	\,.
$$
This implies that if $k \geq \frac{n-1}{d+1} + 1$ then $t = n$, so we cannot construct a $(k,n,d)$-superimposed code of length $t$ that is better than the identity matrix of size $n \times n$.
\end{rem}

By Remark \ref{rem:NonExistence},  it is clear that the constraint $k \leq n/e$ in Corollary \ref{cor:impBound2} is reasonable since $1 + (k-1) (d+1) \geq e k$ for every $k, d \geq 2$.

We also note that a simple generalization of the method given by Cheng et al.  in \cite{Cheng} provide the following result.

\begin{thm}\label{th:ChengExtension}
There exists a $(k,n,d)$-superimposed code of length $t$,  $t \leq \frac{1}{B_k} \left(k\log(n/k) + \log(k e^k) \right)$, where
\begin{align}
&B_k = \max_{q \geq 2} B_{k,q}\,, \label{eq:Bk} \\
&B_{k,q} = \frac{- \log \left[1 - \left(1-\frac{1}{q}\right)^{k-1}\right]}{q+d}\,. \nonumber
\end{align}
For $k \to \infty$,  the point $q$ that maximize \eqref{eq:Bk} is linear in $k$.
\end{thm}
The proof of Theorem \ref{th:ChengExtension} is similar to the one in \cite{Cheng}, we only need to ensure that when we construct a binary matrix $M$ starting from a random $q$-ary matrix each column $\bf c$ of $M$ has a run of at least $d$ $0$'s between any two $1$'s. This can be done by mapping each $q$-ary symbol into a binary vector of length $q+d$ where the last $d$ elements are fixed to $0$.

If we lower bound $B_k$ with $B_{k,  q}$ for the choice ${q=\frac{1}{\ln2} (k-1)}$  then,  for $k$ sufficiently large,  Theorem \ref{th:ChengExtension}  gives the following explicit bound on the minimum length $t$ of $(k,n,d)$-superimposed codes
\begin{multline}\label{eq:ChengExplicit}
t \leq d k \log (n/k)   +  \frac{1}{\ln 2} \cdot  k(k-1) \log (n/k) \\ + \left( \frac{1}{\ln 2} (k-1) + d \right) \log\left(k e^k\right) + O(1).
\end{multline}
One can see that this bound already improves, for $k$ sufficiently large, the one given in Theorem \ref{th:OlgicaThm} but not the one obtained in Corollary \ref{cor:impBound2} for $k < d$.

\section{Selectors}

Selectors were introduced in \cite{Vaccaro2} and they can be seen as a generalization of superimposed codes.  Like superimposed codes, selectors find applications in many circumstances, like Group Testing \cite{Vaccaro2},  efficient conflict resolution in the  transmission model of  \cite{KG}, etc.. In this section, we introduce  selectors in which the weight of each column is equal to some fixed value $w$ and where  any two 1's  in each column of $M$  are  separated by a run of  at least $d$ 0's, so that they can be applied to the scenario of \cite{Olgica}. Successively, we will show that selectors can be used to construct efficient two-stage procedure for Group Testing with runlength constraints, that require  a much  smaller number of tests, with respect to the NAGT considered in \cite{Olgica} and in the previous section of the present paper.
 Let us start by giving  some definitions.

\begin{defn}\label{def:selectors} 
Let $k$, $n$, $d$, $p$ be positive integers,  $1 \leq p \leq k \leq n$.
A $(k, n,d,p)$-\emph{selector} is a $t \times n$ binary matrix $M$ such that  any two 1's  in each column of $M$  are  separated by a run of  at least $d$ 0's, and for any $k$-tuple of the columns of $M$ we have that 
at least $p$ rows of the identity matrix of size $k \times k$ are contained in that $k$-tuple of columns.  The number of rows $t$ of $M$ is called the length of the $(k,n,d, p)$-selector.
\end{defn}

One can see that for $p=k$ we get the definition of $(k,n,d)$-superimposed codes studied  in Section \ref{sec:superimposed}.

\begin{defn}
A $(k,n,d,p, w)$-\emph{selector} is a $(k,n,d,p)$-selector with the additional constraint that each column has weight $w$.
\end{defn}

It can be seen (see \cite[Lemma 2]{Vaccaro1}) that Definition \ref{def:selectors} is equivalent to requiring that for any $k$-tuple of columns of a $(k,n,d,p)$-selector and any $k-p+1$ columns among the selected $k$-tuple, there exists a row of the identity matrix of size $k \times k$  where the $1$ is contained in one of the $k-p+1$ columns. 
Therefore, thanks to this equivalence,  we can generalize the proof of Theorem \ref{th:LLExistence} to obtain the following.

\begin{thm}\label{th:LLSelectors}
There exists a $(k,n,d,p,w)$-selector of length $t$,  where $t$ is the minimum integer such that the following inequality holds
\begin{multline}\label{eq:LLSelBound}
e  \binom{k}{p-1} \left[ \binom{n}{k} - \binom{n-k}{k} \right]\cdot \\ \left( \frac{w (k-1) - \frac{w-1}{2}}{t - (w-1) d - \frac{w-1}{2}} \right)^{w(k-p+1)} \leq 1.
\end{multline}
\end{thm}

\begin{proof}
Let $M$ be a $t \times n$ binary matrix, where each column $\mathbf{c}$ is picked uniformly at random between the set of all distinct binary vectors of length $t$ such that each column has weight $w$ and with distance between ones in each column at least $d$.  As in Theorem \ref{th:LLExistence},  by Lemma \ref{lem:enumerative} we have that
$$\Pr(\mathbf{c}) = \binom{t - (w-1) d}{w}^{-1}\,.$$
For a given pair of sets $B_1$, $B_2 \subseteq [1, n]$ where $|B_1| = k-p+1$, $|B_2| = p-1$ and $B_1 \cap B_2 = \emptyset$,  let $E_{B_1, B_2}$ be the event such that for each column $\mathbf{c}_i$ with $i \in B_1$ and every row $r$ where $\mathbf{c}_i (r) = 1$ there exists an index $j \in (B_1 \cup B_2 \setminus \{i\})$ such that $\mathbf{c}_j$ has $1$ in that same row $r$.  There are $\binom{k}{p-1} \binom{n}{k}$ such events.  Then,  by the same argument  used in the proof of Theorem \ref{th:LLExistence} we can easily upper bound the probability of such events as follows
\begin{equation}\label{eq:probUpperSel}
	\Pr(E_{B_1, B_2}) \leq \left( \frac{\binom{w(k-1)}{w}}{\binom{t-(w-1)d}{w}}\right)^{k-p+1}\,.
\end{equation}
Using Corollary \ref{cor:boundBinomial} we upper bound \eqref{eq:probUpperSel} as follows
\begin{equation}\label{eq:probUpperSel2}
\Pr(E_{B_1, B_2}) \leq  \left( \frac{w (k-1) - \frac{w-1}{2}}{t - (w-1) d - \frac{w-1}{2}} \right)^{w(k-p+1)}\,.
\end{equation}
Let us fix an arbitrary event $E_{A_1, A_2}$ then it is easy to see that it is mutually independent from all the events $E_{A'_1, A'_2}$ such that $A'_1 \subseteq [1,n] \setminus (A_1 \cup A_2)$,  $A'_2 \subseteq [1,n] \setminus (A_1 \cup A_2 \cup A'_1)$.  The number of events $E_{A'_1, A'_2}$ is equal to $\binom{k}{p-1} \binom{n-k}{k}$.  Therefore each event $E_{A_1, A_2}$ is dependent of at most
\begin{equation}\label{eq:DSel}
	f = \binom{k}{p-1} \left[ \binom{n}{k} - \binom{n-k}{k} \right]
\end{equation}
other events.  If the probability that none of the events $E_{A_1,A_2}$ occurs is strictly positive then there exists a matrix $M$ that is a $(k,n,d,p,w)$-selector of length $t$.  Using Lemma \ref{lem:LLL} and taking $p$ equal to the RHS of \eqref{eq:probUpperSel2} and $f$ as defined in equation \eqref{eq:DSel},  Theorem \ref{th:LLSelectors} follows.
\end{proof}

\begin{cor}\label{cor:selBoundOnT}
There exists a $(k,n,d,p, w)$-selector of length $t$, where
\begin{multline}\label{eq:boundwithWSel}
	t \leq \Bigg\lceil (w-1) d +  \frac{w-1}{2} + \left(w(k-1) - \frac{w-1}{2}\right) \cdot  \\ \left(e \binom{k}{p-1} \left[ \binom{n}{k} - \binom{n-k}{k} \right]\right)^{\frac{1}{w(k-p+1)}}  \Bigg\rceil \,.
\end{multline}
\begin{proof}
It follows rearranging the terms in equation \eqref{eq:LLSelBound}.
\end{proof}
\end{cor}

Again, by exploiting the  result by Moser and Tardos \cite{MT}, we get  a $O(n^k)$ randomized Las Vegas  algorithm to construct the codes of Corollary \ref{cor:selBoundOnT}

Thanks to Corollary \ref{cor:selBoundOnT} we obtain the following upper bound on the minimum length of $(k,n,d, p)$-selectors.

\begin{cor}\label{cor:selBoundExplicit}
There exists a $(k,n,d, p)$-selector of length $t$ with $k \leq n / e$, where
\begin{multline*}
    t \leq \ln 2 \cdot \frac{d k}{k-p+1}  \log (n/k) \\ + \ln 2 \cdot e^{3+\frac{1}{e}} \frac{k^2}{k-p+1}  \log (n/k) + O(k \log (n/k)) \,.
\end{multline*}
\end{cor}
\begin{proof}
Substituting $w = \frac{k}{k-p+1} \ln (n / k)$ in \eqref{eq:boundwithWSel} and using the well-known inequality $\binom{m}{s} \leq  \left(\frac{e m}{s}\right)^s$, we get
\begin{multline*}
    t \leq \frac{d k}{k-p+1} \ln (n/k)\\ + e \left[ e^{1 + \frac{p}{k}} \left( \frac{k}{p-1} \right)^{\frac{p-1}{k}}  \right]^{\frac{1}{\ln(n/k)}} \frac{k^2}{k-p+1}  \ln (n/k) \\ + O(k \ln (n/k)) \,.
\end{multline*}
Hence Corollary \ref{cor:selBoundExplicit} follows since $p \leq k$, $n \geq e k$ and since the function $x^{1/x}$ takes its maximum at $x = e$. 
\end{proof}

\subsection{Application of $(k,n,d, p)$-selectors to two-stage Group Testing with runlength constraints}
We need the following result, whose proof for "classical" selectors (that is, for selectors without the runlength constraint studied in  this paper) is implicit in the discussion before Theorem 3 of \cite{Vaccaro2}. It is trivial to  see that  the proof carries out also in the present scenario.

\begin{lem}\label{lem:selectorList}
Let $M$ be a $(k,n,d,p)$-selector with $t$ rows, and let ${\bf f}$ be the  $t\times 1$ columns vector obtained by the bitwise OR of at most $q$, $q\leq p-1$, columns ${\mathbf{c}_{i_1}}, \ldots , {\mathbf{c}_{i_{q}}}$ of $M$. Then, apart from  
${\mathbf{c}_{i_1}}, \ldots , {\mathbf{c}_{i_{q}}}$, 
there are at most other $k-q-1$ columns of $M$  whose 1's are in a subset of the positions in which the vector ${\bf f}$ also has 1's.
\end{lem}

Now we proceed as follows. Let $k$ be an upper bound on the number of possible positives in the Group Testing problem. We perform all the tests corresponding to the rows of a  $(2k,n,d, k+1)$-selector $M$, as explained in the introduction. More precisely, the generic $i$-th pool $T_i$, for $i=1, \ldots , t$, contains  all elements $j\in [1,n]$ for which $M_{ij}=1$.
After having performed (in parallel) all tests on pools $T_1, \ldots , T_t$,
we get a "sindrome"  vector ${\bf f}$ (of dimension $t\times 1$) equal to the bitwise OR of the (at most) $k$ columns that correspond to the unknown positive elements. The number of columns of $M$ that are "covered" by ${\bf f}$ (that is, that have their  1's  in a subset of the positions in which the vector ${\bf f}$ also has 1's) is upper bounded  by $2k$ (by Lemma \ref{lem:selectorList}). In other words, there are at most $s\leq 2k$ potentially positive elements, and the true positive are among them. Hence,  one can test individually those $s$ elements to discover the true positives. Altogether, we have used  $t+2k$ tests.

By using Corollary \ref{cor:selBoundExplicit} to estimate   $t$, we get that we can discover all the positive elements by performing  a number of tests upper bounded by a quantity that is  
\begin{equation}\label{eq:optsel}
    2d\ln(n/k)+O(k \ln (n/k)).
\end{equation}
The bound (\ref{eq:optsel})  shows that our two-stage Group Testing algorithm  outperforms both
the NAGT algorithm presented in \cite{Olgica} and also our improved one
given  in the previous section of the present  paper. It is interesting to notice that the  bound (\ref{eq:optsel})
is information-theoretic optimal, for $d=O(k)$, and that this optimality can be achieved by introducing the least amount 
of adaptivity in the testing algorithm.


\end{document}